%% file: paperlncs.tex
\documentclass[a4paper,10pt]{llncs}
\usepackage[utf8]{inputenc}
\usepackage{esvect}
\usepackage{amsmath}
\usepackage{tikz}
\usepackage[verbose]{wrapfig}
\usepackage{stmaryrd}
\usepackage{amssymb}
\usepackage{tikz}
\usepackage{caption}
\usepackage{algorithm2e}
\usetikzlibrary{tikzmark}
   \usetikzlibrary{calc,trees,positioning,arrows,chains,shapes.geometric,%
      decorations.pathreplacing,decorations.pathmorphing,shapes,%
      matrix,shapes.symbols}
\newcommand{\fence}{\mathsf{fence}}
\usepackage{listings}
\usepackage{stmaryrd}
\usepackage{xcolor}
\makeatletter
\newcommand{\removelatexerror}{\let\@latex@error\@gobble}

\makeatother

\title{Efficient Verification of Concurrent Programs Over the TSO Memory Model}
\author{Chinmay Narayan, Subodh Sharma, S.Arun-Kumar}
\date{}
\institute{Indian Institute of Technology Delhi}
\begin{document}
\include{diagram}
\include{defs}
\maketitle

\begin{abstract}
  We address the problem of efficient verification of multi-threaded
  programs running over Total Store Order (TSO) memory model. It has
  been shown that even with finite data domain programs, the
  complexity of control state reachability under TSO is non-primitive
  recursive. In this paper, we first present a bounded-buffer
  verification approach wherein a bound on the size of buffers is
  placed; verification is performed incrementally by increasing the
  size of the buffer with each iteration of the verification procedure
  until the said bound is reached.  For programs operating on finite
  data domains, we also demonstrate the existence of a buffer bound
  $k$ such that if the program is safe under that bound, then it is
  also safe for unbounded buffers. 
  We have implemented this technique in a tool \texttt{ProofTraPar}.
  Our results against \texttt{memorax} \cite{memoraxtacas11}, a
  state-of-the-art sound and complete verifier for TSO memory model,
  have been encouraging.
\end{abstract}

\input{intro}
\input{related}
\input{prelim}

\input{overview}
\input{fence}
\input{experiments}
\input{conclusion}
\bibliography{paperlncs}
\bibliographystyle{acm}
\include{appendix}
\end{document}

%% file: diagram.tex
\tikzset{
      rect/.style={
      rectangle,
      very thick,
      draw=red!50!black!50, 
      top color=white, 
      bottom color=red!50!black!20, 
      font=\itshape
      },
     forallnode/.style={
      rectangle,
      minimum size=6mm,rounded corners=2mm,
      very thick,
      draw=red!50!black!50, 
      top color=white, 
      bottom color=red!50!black!20, 
      font=\itshape
      },
}
\tikzstyle{every node}=[draw=white,thick,anchor=west]
\tikzstyle{selected}=[draw=red,fill=red!30]
\tikzstyle{optional}=[dashed,fill=gray!50]
\newcommand{\stateofart}{
\begin{tikzpicture}[ grow via three points={one child at (0.5,-0.7) and
  two children at (0.5,-0.7) and (0.5,-1.4)},
  edge from parent path={(\tikzparentnode.south) |- (\tikzchildnode.west)}]
   \node {RMM Verification}
    child { node {Precise}
        child { node {Safety Property}
	  child { node{$\begin{array}{l}\mbox{Memorax\cite{Atig:2010:VPW:1707801.1706303,memoraxtacas11}}\\ \mbox{Remmex\cite{DBLP:conf/spin/LindenW11}\cite{remmextacas13}} \end{array}$}}
	  }
	child [missing] {}				
       child { node {SC Property}
	  	child { node{$\begin{array}{l}\mbox{Robustness\cite{ShashaSnirtoplas88,Burnimtestingstabilitytacas11,bouajjanirobustnessesop13,aglavestabilitycav11}}\\ \mbox{Persistence\cite{persistenceesop15}} \end{array}$}}
	}
	child [missing] {}				
    }
     child [missing] {}				
     child [missing] {}
     child [missing] {}				
     child [missing] {}				
    child { node {Under-approximate}
      child { node {Buffer-bounded}
	     	  	child { node{
	     	  	$\begin{array}{l}\mbox{\cite{predicateabstractionsas13,inferencefences13,refinementpropogationsas14,effectiveabstractionvmcai15}}  
	     	  	\end{array}$
	     	  	}
	     }
	     child [missing] {}
	  }
        child [missing] {}
       child { node {Context-bounded}
	  	child { node{$\begin{array}{l}\mbox{\cite{ridofstorebuffercav11}} \end{array}$}}
	}
	child [missing] {}				
     }
     	child [missing] {}				
     		child [missing] {}				
     			child [missing] {}				
    child { node {Over-approximate}
      child { node{$\begin{array}{l}\mbox{\cite{AlglaveKNP14,programtransformationjadeesop13,coherenceabstractionpldi11}} \end{array}$}}
    };
\end{tikzpicture}
}

\newcommand{\tsosemantics}{
\removelatexerror
 \scalebox{0.85}{\parbox{\linewidth}{%
 {
 \begin{minipage}{\linewidth}
\begin{minipage}{0.5\textwidth}
\begin{align}\tag{\textsc{BWrite}}\label{tag:bwrite}
\frac
{\begin{array}{c}\Inst(a)=(x:=e),~ \Adenot{e}{\LM[t]}=v,\\~\Buff'_k=\Buff_k[t\leftarrow \Buff_k(t).(sv,v)], \len{\Buff_k(t)}< k\end{array}}
{(\Q,\LM,\GM,\Buff_k)\stackrel{a}{\to_k}(\Q',\LM,\GM,\Buff'_k)}
\end{align}	
\end{minipage}
\hspace{2.1cm}
\begin{minipage}{0.5\textwidth}
\begin{align}\tag{\textsc{BRead}}\label{tag:bread}
\frac
{\begin{array}{c}\Inst(a)=(\ell:=x),~\LM'=\LM[(t,\ell)\leftarrow v]\\~ \restr{\Buff_k(t)}{\{x\}\times \Val}=\alpha.(x,v)\end{array}}
{(\Q,\LM,\GM,\Buff_k)\stackrel{a}{\to_k}(\Q',\LM',\GM,\Buff_k)}
\end{align}	
\end{minipage}
\end{minipage}
 \begin{minipage}{\linewidth}
\begin{minipage}{0.5\textwidth}
\begin{align}\tag{\textsc{MRead}}\label{tag:mread}
\frac
{\begin{array}{c}\Inst(a)=(\ell:=x),~\LM'=\LM[(t,\ell)\leftarrow \GM(x)], \\~ \restr{\Buff_k(t)}{\{x\}\times \Val}=\epsilon\end{array}}
{(\Q,\LM,\GM,\Buff_k)\stackrel{a}{\to_k}(\Q',\LM',\GM,\Buff_k)}
\end{align}	
\end{minipage}
\hspace{2.1cm}
\begin{minipage}{0.5\linewidth}
\begin{align}\tag{\textsc{LWrite}}\label{tag:mlwrite}
\frac
{\begin{array}{c}\Inst(a)=(\ell:=e),\Adenot{e}{\LM[t]}=v, \\~\LM'=\LM[(t,\ell)\leftarrow v]\end{array}}
{(\Q,\LM,\GM,\Buff_k)\stackrel{a}{\to_k}(\Q',\LM',\GM,\Buff_k)}
\end{align}	
\end{minipage}
\end{minipage}
\begin{minipage}{\linewidth}
\begin{minipage}{0.5\textwidth}
\begin{align}\tag{\textsc{Assume}}\label{tag:massume}
\frac
{\begin{array}{c}\Inst(a)=(\assume{e}),\Adenot{e}{\LM[t]}=\true\end{array}}
{(\Q,\LM,\GM,\Buff_k)\stackrel{a}{\to_k}(\Q',\LM,\GM,\Buff_k)}
\end{align}	
\end{minipage}
\hspace{2.1cm}
\begin{minipage}{0.5\linewidth}
\begin{align}\tag{\textsc{flush}}\label{tag:mflush}
\frac
{\begin{array}{c}\Buff_k(t)=(x,v).\alpha, \Buff'_k=\Buff_k[t\leftarrow \alpha]\end{array}}
{(\Q,\LM,\GM,\Buff_k)\stackrel{a}{\to_k}(\Q',\LM',\GM[sv\leftarrow v],\Buff'_k)}
\end{align}	
\end{minipage}
\end{minipage}
\begin{minipage}{\linewidth}
\begin{center}
\begin{minipage}{0.5\textwidth}
\begin{align}\tag{\textsc{fence}}\label{tag:mfence}
\frac
{\begin{array}{c}\Inst(a)=(\fence), \Buff_k=\epsilon\end{array}}
{(\Q,\LM,\GM,\Buff_k)\stackrel{a}{\to_k}(\Q',\LM,\GM,\Buff_k)}
\end{align}	
\end{minipage}
\vspace{0.5cm}
\end{center}
\end{minipage}
\begin{minipage}{1.15\linewidth}
\caption{TSO semantics for \emph{k-bounded} buffers. All these rules
  assume transitions for process $t$, ie. $\Q[t]=q$, $(q,a,q')\in
  \delta_t$, $\Q'=\Q[t \leftarrow q']$, $\Adenot{e}{\LM[t]}$ denotes the  value of $e$ under the store $\LM[t]$. $\restr{\Buff_k(t)}{\{x\}\times \Val}$ denotes the buffer content of $\Proc{t}$ restricted to variable $x$.}
\label{fig:tsok}
\end{minipage}
}}
}
}

%% file: defs.tex
\newcommand{\TID}{\mathsf{TID}}
\newcommand{\vr}[1]{\mathrm{#1}}
\newcommand{\loc}{\mathsf{Loc}}
\newcommand{\Labl}{\mathsf{LABL}}
\newcommand{\Inst}{\mathsf{Ins}}
\newcommand{\instset}{\mathsf{INST}}
\newcommand{\St}{\mathsf{S}}
\newcommand{\st}{\mathsf{s}}
\newcommand{\Q}{\mathsf{cs}}
\newcommand{\LM}{\mathsf{Lm}}
\newcommand{\GM}{\mathsf{Gm}}
\newcommand{\TSO}{\mathsf{TSO}}
\newcommand{\Buff}{\mathsf{Buff}}
\newcommand{\Val}{\mathsf{Val}}
\newcommand{\Var}{\mathsf{Var}}
\newcommand{\Exp}{\mathsf{Exp}}
\newcommand{\Powerset}[1]{\mathbb{P}({#1})}
\newcommand{\Elem}[1]{\mathcal{EL}(#1)}
\newcommand{\Assign}[2]{\mathrm{#1}\mathbf{:=}\mathrm{#2}}
\newcommand{\subst}[3]{\mathtt{\linebreak[0]#1[\linebreak[0]{#2}/\linebreak[0]{#3}]}}
\newcommand{\w}[2]{!({{#1}},{#2})}
\renewcommand{\r}[2]{?({{#1}},{#2})}
\newcommand{\inarr}[1]{\begin{array}{@{}l@{}}#1\end{array}}
\newcommand{\lab}[1]{{\mathtt{#1.}}~~}
\newcommand{\true}{\mathsf{true}}
\newcommand{\false}{\mathsf{false}}
\newcommand{\set}[1]{\{#1\}}
\newcommand{\restr}[2]{#1\downharpoonleft_{#2}}
\newcommand{\Last}{\mathbf{Last}}
\newcommand{\lland}{~\land~}
\newcommand{\llor}{~\lor~}
\newcommand{\llnot}{\lnot~}
\newcommand{\limp}{\Rightarrow}
\newcommand{\aut}[1]{\mathcal{A}(#1)}
\newcommand{\htt}[1]{\widehat{#1}}
\newcommand{\Lab}{\mathtt{Lab}}
\newcommand{\ndelta}{\mathcal{\delta}}
\newcommand{\unsatcore}{\mathtt{Unsatcore}}
\newcommand{\validcore}{\mathtt{validcore}}
\newcommand{\A}[1]{\mathcal{A}(#1)}
\newcommand{\len}[1]{|{#1}|}
\newcommand{\HOARE}[3]{\colorbox{white}{${\{{#1}\}~ \linebreak[0] #2 ~\linebreak[0]\{{#3}\}}$}}
\newcommand{\word}[1]{\mathtt{``#1"}}
\newcommand{\tafap}{B_{(\sigma, \neg \phi)}}
\newcommand{\bld}[1]{\mathbf{#1}}
\newcommand{\lang}[1]{\mathcal{L}(#1)}
\newcommand{\HMap}{\mathtt{HMap}}
\newcommand{\acc}{\mathtt{acc}}
\newcommand{\faS}{S_{\forall}}
\newcommand{\teS}{S_{\exists}}
\newcommand{\AssnMap}{\mathtt{AMap}}
\newcommand{\ResdMap}{\mathtt{RMap}}
\newcommand{\rev}{\mathsf{rev}}
\newcommand{\rf}{\mathrm{rf}}
\newcommand{\run}{\sigma}
\newcommand{\PH}{\mathsf{PH}}
\newcommand{\prog}[6]{\langle #1,\linebreak[0] #2,\linebreak[0] #3,\linebreak[0] #4, \linebreak[0]#5, \linebreak[0]#6\rangle}
\newcommand{\Proc}[1]{P_{#1}}
\newcommand{\OP}{\mathrm{op}}
\newcommand{\lcas}[3]{\mathtt{lock}({#1})}
\newcommand{\assume}[1]{\mathtt{assume}(\mathrm{#1})}
\newcommand{\assrt}[1]{\mathtt{assert}(\mathrm{#1})}
\newcommand{\SV}{\mathtt{SV}}
\newcommand{\LV}{\mathtt{LV}}
\newcommand{\p}{\mathit{p}}
\newcommand{\nop}{\mathtt{nop}}
\newcommand{\partialfun}{\hookrightarrow}
\newcommand{\AssnM}{\mathrm{Assrn}}
\newcommand{\BExp}{\mathtt{BExp}}
\newcommand{\succset}{\mathtt{succ}}
\newcommand{\track}[1]{\colorbox{yellow}{\textcolor{red}{#1}}}
\newcommand{\defeq}{\stackrel{def}{=}}
\newcommand{\afa}[1]{\mathcal{\hat{A}}_{#1}}
\newcommand{\lb}[1]{\mathtt{#1}}
\newcommand{\Adenot}[2]{\llbracket{#1}\rrbracket_{#2}}
\newcommand{\TSOr}[1]{\mathsf{TSO^{\sharp}}_{#1}}
\newcommand{\Str}{\mathsf{S}{^{\sharp}}}
\newcommand{\str}{\mathsf{s}{^{\sharp}}}
\newcommand{\tor}[1]{\Rightarrow_{#1}}
\newcommand{\Buffr}{{\mathsf{Buff}^{\sharp}}}
\newcommand{\LI}{\mathsf{Li}}
\newcommand{\sigmar}{\sigma^{\sharp}}
\newcommand{\flush}{\mathsf{flush}}
\newcommand{\FV}{\mathsf{FV}}
\newcommand{\SCI}{\mathsf{SCI}}
\newcommand{\und}{\mathsf{Undef}}
\newcommand{\GMr}{\GM{^{\sharp}}}
\newcommand{\LMr}{\LM{^{\sharp}}}
\newcommand{\AMap}{\mathsf{AMap}}
\newcommand{\Assrn}{\mathsf{Assrn}}
\newcommand{\comp}[2]{#1 \! \circ \! #2}

%% file: intro.tex
\section{Introduction}\label{sec:intro}
The explosion in the number of schedules is central to the complexity
of verifying the safety and correctness of concurrent programs.
There exist a plethora of approaches in the literature that explore
ways and means to address the schedule-space explosion problem; 
incidentally, many of the the published techniques operate
over the assumption of a sequentially consistent (SC) memory model. In
contrast, almost all modern multi-core processors conform to memory
models \emph{weaker} than SC. A program executing on a relaxed memory
model exhibits more behaviours than on the SC memory model. As a
result, a program declared correct by a verification methodology that
assumes SC memory model can possibly contain a buggy behaviour when
executed on a relaxed memory model. 

Consider x86 machines that conform to TSO (Total Store Ordering).  The
compiler or the runtime system of the program under the TSO memory
model is allowed to reorder a read following a write (read and writes
are to different variables) within a process, i.e. break the
program order specified by the developer. Operationally, such a
re-ordering is achieved by maintaining per-process store
buffers. Write operations issued by a process/thread are enqueued in
the store buffer local to that process.  The buffered writes are later
flushed (from the buffer) into the global memory. The point in time
when flushes take place is deterministically known only when the store
buffers are full. When the store buffers are partially full, flushes
are allowed to take place non-deterministically. Therefore, when a
read operation of variable $\vr{x}$, is executed by a process, the
process first checks whether there is a recent write to $\vr{x}$ in
the process's store buffer. If such a write exists then the value from
store buffer is returned, otherwise the value is read from the global
memory.

\begin{wrapfigure}{l}{0.6\textwidth}
%
\removelatexerror
\scalebox{0.75}{\parbox{\linewidth}{%
    \centering{$\mathrm{flag_1}=\false,\mathrm{flag_2}=\false,\mathrm{t}=0$;
      \normalsize}\\ $\begin{array}{l@{~~}|@{~~}l} \footnotesize{
        \inarr{ {}\\ ~~~~~~
          P_{1}~~~~~\\ \mathrm{While}(\true)\{\\ \lab{1}
          \Assign{flag_1}{\true};\\ \lab{2} \Assign{t}{2};\\ \lab{3}
          \mathrm{while}(\mathrm{flag}_2=\mathrm{true} \,\&\,
          \mathrm{t}=2);\\ \lab{4} \bld{ //Critical~Section}\\ \lab{5}
          \Assign{flag_1}{\false};\\ \} } }& \footnotesize{ \inarr{
          {}\\ ~~~~~~ P_{2}~~~~~\\ \mathrm{While}(\true)\{\\ \lab{6}
          \Assign{flag_2}{\true};\\ \lab{7} \Assign{t}{1};\\ \lab{8}
          \mathrm{while}(\mathrm{flag_1}=\mathrm{true} \,\&\,
          \mathrm{t}= 1);\\ \lab{9}
          \bld{//Critical~Section}\\ \lab{10} \Assign{flag_2}{\false};
          \\ \} } }
\end{array}$
\caption{Peterson's algorithm for two processes }
\label{fig:pet}
}}
\end{wrapfigure}
Figure \ref{fig:pet} shows Peterson's algorithm as an instance of a
correct program under SC semantics but which can fail when executed under
TSO.
In this algorithm, two processes $P_{1}$
 and $P_{2}$ coordinate their access to their respective critical sections using a shared variable
$\mathrm{t}$. This algorithm satisfies the mutual exclusion property under
the SC memory model, i.e. both processes can not be simultaneously
present in their critical sections. 
The property however, is violated when the same algorithm
is executed with a weaker memory model, such as TSO.
Consider the following execution under TSO. The write operations at
$\lb{1}$, $\lb{2}$, $\lb{6}$ and $\lb{7}$ from processes $\Proc{1}$
and $\Proc{2}$ are stored in store buffers and are yet to be reflected
in the global memory. The reads at control locations $\lb{3}$ and
$\lb{8}$ will return initial values, thereby violating the mutual
exclusion property.

One can avoid such erroneous behaviors and restore the SC semantics of
the program by inserting special instructions, called \emph{memory
  fence}, at chosen control locations in the program.  A memory fence
ensures that the store buffer of the process (which executes the fence
instruction) is flushed entirely before proceeding to the next
instruction for execution. In the  example, when
\emph{fence} instructions are placed after $\Assign{flag_1}{\true}$ in
$\Proc{1}$ and after $\Assign{flag_2}{\true}$ in $\Proc{2}$, the
mutual exclusion property is restored.

Safety verification under TSO is a hard problem even in the case
of finite data domain programs. The main reason for this 
complexity is the unboundedness
of store buffers. A program can be proved correct under TSO only when
the non-reachability of the error location is shown irrespective of
the bound on the buffers.  The work in
\cite{Atig:2010:VPW:1707801.1706303} demonstrated the equivalence of
the TSO-reachability problem to the coverability problem of lossy
channel machines which is decidable and of non-primitive recursive
complexity. A natural question is to ask if it is possible to have a
buffer bound $k$ such that if a finite data domain program is safe
under the $k$-bounded TSO semantics then it is guaranteed to be safe
even with unbounded buffers. For programs without loops such a
statement seems to hold intuitively. For programs with loops, it is
possible that a write instruction inside a loop keeps filling 
the buffer with values without ever getting them flushed to the main
memory. However, for finite data domain programs, only a finite set of
different values will be present in this unbounded buffer and this
leads to a sufficient bound on the buffer size.

In this paper we show that it is possible to verify a program $P$
under TSO (with unbounded buffers) by generalizing the bounded buffer
verification. Towards this we first define $\TSO_k$, TSO semantics
with buffer size $k$, and then characterize a bound $k_0$ such that if
a program is safe in $\TSO_{k_0}$ then it is safe for any buffer bound
greater than $k_0$. We adapt a recently proposed trace
partitioning based approach
\cite{Farzan:2013:IDF:2429069.2429086,DBLP:journals/corr/NarayanGA15}
for the TSO memory model. These methods work for the SC memory model
as follows: the set of all SC executions of a program $P$ are
partitioned in a set of equivalence classes such that it is sufficient
to prove the correctness of only one execution per equivalence
class. As this trace partitioning approach works with symbolic
executions, we first define an equivalent TSO semantics to generate a
set of symbolic TSO traces. Subsequently we invoke a trace
partitioning tool \texttt{ProofTraPar}
\cite{DBLP:journals/corr/NarayanGA15} for proving the correctness of
these traces. Note that the set of behaviors of a program $P$ under
$\TSO_k$ is a subset of the behaviors of $P$ under $\TSO_{k'}$ for any
$k'>k$. The trace partitioning approach allows us to reuse
the proof of correctness of $P$ with buffer bound $k$ in the proof of
correctness of $P$ with any buffer bound greater than $k$. In a
nutshell, the main contributions of this work are:
\begin{itemize}

 \item We characterize a buffer bound in case of finite state programs
   such that if the program is correct  under TSO  up to
   that bound then it is correct for unbounded buffers as well.   
 \item We adapt the recently proposed trace partition based proof
   strategy of SC verification
   \cite{Farzan:2013:IDF:2429069.2429086,DBLP:journals/corr/NarayanGA15}
   for TSO  by defining an equivalent TSO semantics to
   generate a set of symbolic TSO traces.
 \item We implement our approach in a tool,
   \texttt{ProofTraPar}\cite{DBLP:journals/corr/NarayanGA15}, and
   compare its performance against
   \texttt{memorax}\cite{memoraxtacas11}, a sound and complete
   verifier for safety properties under TSO. We perform
   competitively in terms of time as well as space.
   In a few examples, \texttt{memorax} timed out after
   consuming around 6GB of RAM whereas our approach could analyze the
   program in less than 100 MB memory.
\end{itemize}

Section \ref{sec:related} covers the related work in the area of
verification under relaxed memory models. Section \ref{sec:prelim}
covers the notations used in this paper. Section \ref{sec:bound} shows
the necessary and sufficient conditions to generalize bounded
verification to unbounded buffers for finite data domain
programs. Section \ref{sec:tracepart} presents an equivalent TSO
semantics to generate a set of symbolic traces which can be used by the
trace partitioning tool \texttt{ProofTraPar} to check the
correctness under a buffer bound. This section ends with an approach
based on \emph{critical cycle} to insert \emph{memory fence}
instructions. Section \ref{sec:experiment} compares the performance of
our approach with \texttt{memorax}. Section \ref{sec:conclusion}
concludes with future directions. 
%
%
%
%

%% file: related.tex
\section{Related Work}\label{sec:related}
Figure \ref{fig:related} captures the related work in this
area. Verification approaches for relaxed memory models can be broadly
divided into three classes: precise, under-approximate and
over-approximate. For finite state programs, the work in
\cite{Atig:2010:VPW:1707801.1706303,memoraxtacas11}  present sound
and complete algorithms for control state reachability (finite state
programs) under TSO and PSO memory models.

\begin{wrapfigure}{l}{0.5\textwidth}
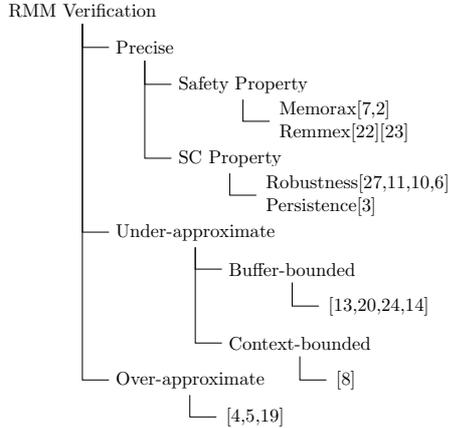
 \removelatexerror
\scalebox{0.7}{\parbox{\linewidth}{\stateofart
 \caption{State of the art}
 \label{fig:related} }} 
\end{wrapfigure}
Sets of infinite configurations, arising from unbounded
buffer size, are finitely presented using regular
expressions. Acceleration based techniques that led to faster
convergence in the presence of loops were presented in
\cite{DBLP:conf/spin/LindenW11,remmextacas13}. However, the
termination of the algorithm was not guaranteed. Notice that in both
\cite{memoraxtacas11} and \cite{remmextacas13} the specification was a
set of control states to be avoided. One can also ask the state
reachability question with respect to SC specification, i.e. does a
program $P$ reach only SC reachable states under a relaxed memory
model? This problem was shown to be of the same complexity as of SC
verification (Pspace-complete) and hence gave a more tractable
correctness criterion than general state reachability
problem. \cite{ShashaSnirtoplas88,Burnimtestingstabilitytacas11,bouajjanirobustnessesop13,aglavestabilitycav11,persistenceesop15}
work with this notion of correctness and give efficient algorithms to
handle a range of memory models. In this paper we work with the
control state reachability problem as opposed to the SC state
reachability problem.

Over-approximate analyses
\cite{AlglaveKNP14,programtransformationjadeesop13,coherenceabstractionpldi11}
trade precision with efficiency and construct an over-approximate set
of reachable states. Recently
\cite{DBLP:conf/tacas/AbdullaAAJLS15,Zhang:2015:DPO:2737924.2737956}
used stateless model checking under TSO and PSO memory models. The
main focus of these approaches are in finding bugs rather than proving
programs correct. Another line of work to make the state reachability
problem more tractable involved either restricting the size of buffers
\cite{predicateabstractionsas13,inferencefences13,refinementpropogationsas14,effectiveabstractionvmcai15}
or bounding the context switches \cite{ridofstorebuffercav11} among
threads. None of theses methods give completeness guarantee even for
the finite data domain programs.

%% file: prelim.tex
\section{Preliminary}\label{sec:prelim}

A concurrent program is a set of processes uniquely
identified by indices $t$ from the set $\TID$. As in
\cite{memoraxtacas11,DBLP:conf/fase/BouajjaniCDM15}, a process
$\Proc{t}$ is specified as an automaton $\langle
Q_t,\Labl_t,\delta_t,q_{0,t}\rangle$. Here $Q_t$ is a finite set of
control states, $\delta_t \subseteq Q_t \times \Labl_t \times Q_t$ is
a transition relation and $q_{0,t}$ is the initial state. Without loss
of generality we assume every transition is labeled with a different
symbol from $\Labl_t$.  $\Labl_t$ represents a finite set of labels to
symbolically represent the instructions of the program.  Let $\SV$ be
the set of shared variables of program $P$ ranged over by $x,y,z$,
$\Val$ be a finite set of constants ranged over by $v$, $\LV_t$ be the
set of local variables of process $\Proc{t}$ ranged over by $\ell,m$,
and $\Exp_t$ be the set of expressions constructed using $\LV_t$,
$\Val$ and appropriate operators. Let $\LV=\bigcup_t \LV_t$,
$\Exp=\bigcup_t \Exp_t$, and $\Labl=\bigcup_t\Labl_t$. Let
$\lb{a},\lb{b},\lb{c}$ range over $\Labl$ and and $e$ range over
$\Exp$. Formally an instruction, from set $\instset$, is one of the
following type; (i) $\Assign{x}{e}$, (ii) $\Assign{\ell}{x}$, (iii)
$\Assign{\ell}{e}$, (iv) $\assume{e}$, and (v) $\fence$, where
$\vr{x}\in \SV$, $\ell\in \LV$ and $\mathrm{e}\in \Exp$.  A function
$\Inst:\Labl\to \instset$ assigns an instruction to every label.  

The first two assignment instructions, (i) and (ii), are the write and
the read operations of shared variables, respectively. Instruction
(iii) assigns the value of an expression (constructed from local
variables and constants) to a local variable, hence, does not
include any shared memory operation. Instruction (iv) is used to model
loop and conditional statements of the program. Note that the boolean
expression $e$ in $\assume{e}$ does not contain any shared
variable. Instruction (v) represents the fence operation provided by
the TSO architecture. Let $\loc(\lb{a})$ be the shared variable used
in $\Inst(\lb{a})$. For a function $\mathsf{F:A\times B}$, let the
function $\mathsf{F}[p\leftarrow q]$ be the same as $\mathsf{F}$
everywhere except at $p$ where it maps to $q$. 

\paragraph{TSO Semantics} In the TSO memory model, every process has a
buffer of unbounded capacity. However, we present the TSO semantics by
first defining a \emph{k-bounded} TSO semantics where all buffers are
of fixed size $k$. For a concurrent program $P$, the \emph{k-bounded}
semantics is given by a transition system $\TSO_k=\langle
\St,\to_k,\st_0 \rangle$. Every state $\st\in \St$ is of the form
$(\Q,\LM,\GM,\Buff_k)$ where process control states $\Q:\TID\to Q$,
$Q=\bigcup_t Q_t$, local memory $\LM:\TID \times \LV \to \Val$, global
memory $\GM:\SV \to \Val$, and \emph{k-length} bounded buffers
$\Buff_k:\TID\to (\SV\times \Val)^k$.  
We overload operator `.' to denote the concatenation of labels as well 
as a dereferencing operator to identify a specific field inside a state.
Therefore, for a state $s$, $s.\GM$, $s.\LM$ and $s.\Buff_k$ denote the 
functions representing global memory, local memory, and buffers respectively. 
Every write operation to a shared variable by process $\Proc{t}$ initially
gets stored in the process-local buffer provided that the buffer has less than
\emph{k} (buffer-bound) elements. This write operation is
later removed from the buffer non-deterministically to update the
global memory. A read operation of a shared
variable say $\vr{x}$, by a process $\Proc{t}$ first checks the local
buffer for any write to $\vr{x}$. If buffer contains any write to
$\vr{x}$ then the value of the last write to $\vr{x}$ is returned as a
result of this read operation. If no such write is
present in the buffer of $\Proc{t}$ then the value is read from the
global memory. A process executes instruction
$\fence$ only when its local buffer is empty. For
instruction $\assume{e}$, boolean expression $\mathrm{e}$ is evaluated
in the local state of $\Proc{t}$. Execution proceeds only when the
expression $\mathrm{e}$ evaluates to $\true$. 
Assignment operation involving only local
variables changes the local memory of $\Proc{t}$. 
The transition relation $\to_k$ is defined in
detail in the Appendix.  

\vspace{-0.1cm}
\paragraph{Relevance of the buffer size $k$} Parameter $k$ influences
the extent of reordering that happens in an execution. For example, if
$k=0$ then no reordering happens and the set of executions is the same
as under the SC memory model. Size parameter $k$, under the TSO memory
model, allows any two instructions separated by at most $k$
instructions to be reordered, provided that one is write and another
is a read instruction. This \emph{reorder-bounded analysis} was also
shown effective by \cite{DBLP:conf/fm/0001K15} and seems a natural way
to make this problem tractable.

\input{tsodecide}

%% file: tsodecide.tex
\newcommand{\dom}[1]{\mathsf{D}_{#1}} \newcommand{\LW}{\mathsf{LW}}
\newcommand{\lbr}{\linebreak[1]}
\newcommand{\Exec}[1]{\mathsf{Exec}({\linebreak[0]#1})}
\newcommand{\nrestr}[2]{{#1{_{\downharpoonleft{_{#2}}}}}}
\newcommand{\ran}{\mathsf{Range}} \newcommand{\eq}{\mathsf{eq}}
\newcommand{\concat}{.}  
\section{Unbounded Buffer Analysis}\label{sec:bound} 
In this section we show that for any finite data domain
program and safety property $\phi$, there exists a buffer size $k_0$
such that it is sufficient to prove $\phi$ for all buffers up-to size
$k_0$. Note that for programs with write instructions inside loops, it
is possible to keep on writing to the buffer without flushing them to
the main memory. However since the data domain is finite, such
instruction are guaranteed to repeatedly write the \emph{same set of
values} to the buffer. It is this repetition that guarantees the
existence of a sufficient bound on the buffer.

The set of states in $\TSO_k$ are monotonic with respect to the buffer
bound, i.e. $\St_{k}\subseteq \St_{k'}$, for all $k\leq k'$. Let
$\nrestr{\st}{(\Q,\GM,\LM,\Buff_{lst})}$ denote the restriction of a
state in $\St$ to only control states, global memory, local memory,
and last writes (if any) to shared variables in buffers. Let
$\nrestr{\St{_k}}{(\Q,\GM,\LM,\Buff_{lst})}=\{\nrestr{\st}{\Q,\GM,\LM,\Buff_{lst}}\mid
\st\in \St_k\}$ be the states of $\St_k$ after projecting out the
above information. For finite data domain the set
$\bigcup_{k=0}^\infty \nrestr{\St_k}{\Q,\GM,\LM,\Buff_{lst}}$ is
finite because only finitely many different possibilities exist for
functions $\Q$, $\GM$, $\LM$ and $\Buff_{lst}$. Further,
$\nrestr{\St_k}{\Q,\GM,\LM,\Buff_{lst}} \subseteq
\nrestr{\St_{k+1}}{\Q,\GM,\LM,\Buff_{lst}}$. Therefore there exists a
$k_0$ such that $\nrestr{\St_{k_0}}{\Q,\GM,\LM,\Buff_{lst}}$ is equal
to the set $\nrestr{\St_{k_0+1}}{\Q,\GM,\LM,\Buff_{lst}}$. In this
section we show that for every $k>k_0$, sets
$\nrestr{\St_{k}}{\Q,\GM,\LM,\Buff_{lst}}$ and
$\nrestr{\St_{k+1}}{\Q,\GM,\LM,\Buff_{lst}}$ are equal and hence we
can stop the analysis at $k_0$.  

For a buffer $\Buff(t)$, let $\sigma_{\Buff(t).lst}$ denote the
sequence of last writes to shared variables in buffer $\Buff(t)$. Let
$\Exec{\GM,\LM(t),\Buff(t),\sigma.a.\sigma',(x,v),\LM'(t)}$ be a
predicate, where $a=(x:=e)$, that holds true iff (i) after executing
sequence $\sigma_{\Buff(t).lst}.\sigma.a.\sigma'$ from the global
memory $\GM$ and local memory $\LM(t)$ the local memory of process $t$
is $\LM'(t)$ and (ii) in the same sequence the value of expression $e$
in write instruction $x:=e$ at label $a$ is $v$. The following two
lemmas relate the states of $\TSO_k$ and $\TSO_{k+1}$ transition
systems. We use $\st_0\stackrel{\sigma}{\to} \st$ to denote a sequence
of transitions over a sequence $\sigma$ of labels.

\begin{lemma}\label{lem:relatektok+1} For all $n$, $\sigma$, $\st\in
\St_{k+1}$ such that $\st_0\stackrel{\sigma}{\to} \st$,
$\len{\sigma}=n$ and $k\ge 0$, there exists a state $\st'\in \St_{k}$
such that $\st'.\GM=\st.\GM$ and for all $t\in \TID$,
 \begin{enumerate}
 \item \label{lem:p1} $(\len{\st.\Buff_{k+1}(t)}=k+1) \limp\exists
x,v.~$ \\$\left\{\begin{array}{l}(i)\
\st.\Buff_{k+1}(t)=\st'.\Buff_k(t)\concat(x,v)\\ (ii)\ \exists \mbox{
$q_t,q_{t'},\sigma',\sigma''$ such that } (q_t,a,q_{t'})\in \delta_t
\mbox {,} \\ \st'.\Q(t)\stackrel{\sigma'}{\to}q_t,
q_{t'}.\Q(t)\stackrel{\sigma''}{\to}\st.\Q(t),\\ \mbox{$\sigma'$ and
$\sigma''$ do not modify the buffer of }\Proc{t}, \mbox{ and }\\
\Exec{\st'.\GM,\st'.\LM(t),\st'.\Buff_k(t),\sigma'.a.\sigma'',(x,v),\st.\LM(t)}
 						\end{array}\right.$\\
and
   \item \label{lem:p2} $(\len{\st.\Buff_{k+1}(t)}<k+1) \limp$\\
\hspace{1cm}$\left\{\begin{array}{l}(i)\
\st.\Buff_{k+1}(t)=\st'.\Buff_k(t),\\ (ii)\ \st.\LM(t)=\st'.\LM(t),
\mbox{ and }\\(iii)\ \st.\Q(t)=\st'.\Q(t)
 						\end{array}\right.$
 						
  \end{enumerate}
 \end{lemma}
$\st'.\Q(t)=\st.\Q(t)$.
The above lemma states that every state in $\St_{k+1}$
where the buffer sizes of all processes are less than $k+1$, is also
present in $\St_k$.  
The detailed proof of this lemma is given in the Appendix.  
%
Now we are ready to prove that after $k_0$, any increase in buffer
size does not yield any new reachable control location.
\begin{theorem}\label{lem:fp} For all $k$,
  \[
  \begin{array}{lc} (\nrestr{\St_k}{(\Q,\GM,\LM,\Buff_{lst})} =
\nrestr{\St_{k+1}}{(\Q,\GM,\LM,\Buff_{lst})}) &\limp\\
(\nrestr{\St_{k+1}}{(\Q,\GM,\LM,\Buff_{lst})} =
\nrestr{\St_{k+2}}{(\Q,\GM,\LM,\Buff_{lst})})
  \end{array}
  \]
\end{theorem}
\begin{proof} 
there exists a state $\st'\in \St_{k+1}$ such that $\st.\Q=\st'.\Q$,
$\st.\GM = \st'.\GM$, $\st.\LM=\st'.\LM$ and
$\st.\Buff_{lst}=\st'.\Buff_{lst}$.  It is sufficient to show that
$(\nrestr{\St_{k+2}}{(\Q,\GM,\LM,\Buff_{lst})} \subseteq
\nrestr{\St_{k+1}}{(\Q,\GM,\LM,\Buff_{lst})})$ as the other side of
inclusion holds. Let us prove it by contradiction, i.e. there is a
state $\st\in \St_{k+2}$ such that no state $\st'\in \St_{k+1}$ exists
with $\st.\Q=\st'.\Q$, $\st.\GM=\st'.\GM$, $\st.\LM=\st'.\LM$ and
$\st.\Buff_{lst}=\st.\Buff_{lst}$.  Following Lemma
\ref{lem:relatektok+1}, this state $\st$ must have at least one buffer
with $k+2$ entries in it. Without loss of generality let $t\in \TID$
such that $\st.\Buff_{k+2}(t)$ is the only full buffer.
 \begin{enumerate}
  \item Clearly, there exists a state $\st'\in \St_{k+2}$ where all
buffers except $t$ are the same as in $\st$, $\st'.\Buff_{k+2}(t)$ is
of size $k+1$ and there exists a sequence of transitions
$\sigma.a.\sigma'$ from $\st'.\Q(t)$ to $\st.\Q(t)$ by process $t$
with only one write operation $a$.
  \item \label{pnt} For state $\st'$, the conditions $\st'.\GM=\st.\GM$ (as no
flush operation in $\sigma$), and $\mathsf{Exec}(\st'.\GM,\lbr
\st'.\LM(t),\lbr \st'.\Buff_{lst}(t),\lbr \sigma.a.\sigma',\lbr
\st.\LM(t))$ hold.
  \item As all buffers of $\st'$ are of size at most $k+1$ therefore
$\st'$ also exists in $\St_{k+1}$ (Lemma \ref{lem:relatektok+1}).
  \item \label{p3} As $\nrestr{\St_k}{(\Q,\GM,\LM,\Buff_{lst})} =
\nrestr{\St_{k+1}}{(\Q,\GM,\LM,\Buff_{lst})}$ holds, therefore there
exists a state $\st''\in \St_{k}$ such that (i) $\st''.\GM=\st'.\GM$,
(ii) $\st''.\LM=\st'.\LM$, (iii) $\st''.\Q=\st'.\Q$, and (iv)
$\st''.\Buff_{lst}(t)=\st'.\Buff_{lst}(t)$ for all $t\in \TID$.
  \item This state $\st''$ can have at most $k$ entries in its process
buffers. Therefore this state must be present in $\St_{k+1}$ as well.
  \item Using Point \ref{pnt} and the conditions (i),(ii),(iii), and (iv) of Point
\ref{p3} above, we get $\mathsf{Exec}(\st''.\GM,\lbr \st''.\LM(t),\lbr
\st''.\Buff_{lst}(t),\lbr \sigma.a.\sigma',\lbr \st.\LM(t))$. This
implies that after executing the sequence $\sigma.a.\sigma'$ by
process $t$ from state $\st''$ in $\St_{k+1}$ the resultant state, say
$\st'''$ will have at most $k+1$ write entries in the buffer of
process $t$. Further the global memory, local memories, control states
and last writes to shared variables in buffers will be identical in
$\st'''$ and $\st$. Therefore $\st'''\in \St_{k+1}$ is the matching
state with respect to $\st$, a contradiction.
 \end{enumerate}
\end{proof}

%% file: overview.tex
\section{Trace partitioning approach}\label{sec:tracepart}
As a consequence of Theorem \ref{lem:fp} one can use an explicit state
model checker for state reachability analysis of finite data domain
programs. However, in this paper we are interested in adapting a
recently proposed trace partitioning based verification method
\cite{Farzan:2013:IDF:2429069.2429086,DBLP:journals/corr/NarayanGA15}
for relaxed memory models. This method has been shown very effective
for verification under the SC memory model. The approach for SC
verification, as given in \cite{DBLP:journals/corr/NarayanGA15}, is
presented in Algorithm \ref{fig:algo:partition}. Firstly, an automaton
is built that represents the set of symbolic traces under the SC
memory model. For SC memory model such an automaton is obtained by
language level shuffle operation
\cite{Riddle:1979:ASS:2245746.2245934,Hopcroft:2001:IAT:568438.568455}
on individual processes. Subsequently, a symbolic trace is picked from
this automaton and checked against a given safety property using
weakest precondition axioms \cite{Dijkstra:1975:GCN:360933.360975}. If
this trace violates the given property then we have a concrete
erroneous trace. Otherwise, an alternating finite automaton (AFA)
\cite{Chandra:1981:ALT:322234.322243} is constructed from the proof of
correctness of this trace.

 The AFA construction algorithm ensures that
every trace in the language of this AFA is correct and hence can be
safely removed from the set of all symbolic traces of the input
program. This process is repeated until either all symbolic traces are
proved correct or an erroneous trace is found. This algorithm is sound
and complete for finite data domain programs.

\vspace{-0.5cm}
\begin{figure}
 \removelatexerror
\scalebox{0.85}{\parbox{\textwidth}{%
\begin{algorithm}[H]
\SetAlgoLined
\KwIn{A concurrent program $\mathcal{P}=\set{p_1,\cdots,p_n}$ with safety property $\phi$}
\KwResult{$yes$, if program is safe else a counterexample}
Construct the automaton $\aut{\mathcal{P}}$ to capture the set of all SC traces of $P$\;
Let {$\mathtt{tmp}$} be the language of $\aut{\mathcal{P}}$\;
\While{$\mathtt{tmp}$ is not empty}{\label{finalloop}
 \label{four} Let $\run \in \mathtt{tmp}$ with $\phi$ as a safety assertion to be checked\;
 Let $\afa{\run,\neg \phi}$ be the AFA constructed from $\run$ and $\neg \phi$ \label{finalafa}\;
 \uIf{$\run$ violates $\phi$}{\label{finalcounter}
  $\run$ is a valid counterexample\;
  \KwRet($\run$)\label{finalret1}\;
  }\Else{
  {${\mathtt{tmp}}:=\mathtt{tmp}\setminus Rev$}\label{langsub},  where $Rev$ is the reverse of the language of $\afa{\run,\neg \phi}$\;
  }
}
\KwRet($yes$)\label{finalret2}\; 
\caption{SC verification algorithm\cite{DBLP:journals/corr/NarayanGA15}}
\label{fig:algo:partition}
\end{algorithm}}}
\end{figure}
\vspace{-0.5cm}
The main challenge in applying this trace partitioning approach to the TSO memory model is the construction of the set of symbolic traces. Consider a program with two processes in Figure \ref{fig:ex2}. 
With initial values of shared variables $\vr{x}$ and $\vr{y}$ as 0, it is possible to have $\ell_1=\ell_2=0$ under the TSO memory model.
We can construct a symbolic trace $\lb{b.d.a.c}$ such that after
executing this sequence the state $\ell_1=\ell_2=0$ is reached. 

\begin{wrapfigure}[6]{l}{0.3\linewidth}
\removelatexerror
\scalebox{0.85}{\parbox{\linewidth}{%
$\begin{array}{l|l}
   \lab{a} \Assign{x}{1} & \lab{c} \Assign{y}{1}\\
   \lab{b} \Assign{\ell_1}{y} & \lab{d} \Assign{\ell_2}{x}
  \end{array}$
  \caption{}
  \label{fig:ex2}
}}
\end{wrapfigure}

Note that this trace is not constructible using the standard interleaving
semantics which was used to construct the set of traces under the SC
memory model. This is because of the program order between $\lb{a}$
and $\lb{b}$ in process 1 and between $\lb{c}$ and $\lb{d}$ in process
2. To use Algorithm \ref{fig:algo:partition} for the TSO memory model
we would like to first construct a set of all such symbolic traces
such that the sequential executions of these traces yield all
reachable states under the TSO memory model. For the above example, it
involves breaking the program orders $\lb{a}-\lb{b}$ and
$\lb{c}-\lb{d}$ and then applying standard interleaving semantics to
construct symbolic traces under the TSO memory model.
Let us look at another non-trivial example in Figure
\ref{fig:ex3}. 

\begin{wrapfigure}[5]{l}{0.3\linewidth}
\scalebox{0.85}{\parbox{\linewidth}{%
\hspace{2.5cm}$\begin{array}{l|l}
   \lab{a} \Assign{\ell}{2} & \lab{d} \Assign{m}{3}\\
   \lab{b} \Assign{y}{\ell+1} & \lab{e} \Assign{x}{m+2}\\
   \lab{c} \Assign{\ell}{x} & \lab{f} \Assign{m}{y}
  \end{array}$
    \caption{}
    \label{fig:ex3}
  }}
\end{wrapfigure}
Assume the initial values of all variables are 0, and
$\vr{\ell}$, $\vr{m}$ are local variables. In TSO it is possible to
have the final values of variables $\ell$ and $\vr{m}$ as 0. This can
happen when writes at $\lb{b}$ and $\lb{e}$ are still in the buffers
and the read operations at $\lb{c}$ and $\lb{f}$ read from the initial
values. Let us construct a symbolic trace whose sequential execution
will yield this state. In this trace label $\lb{e}$ must appear after
label $\lb{c}$ and label $\lb{b}$ must appear after label $\lb{f}$.
This means that the trace will break either the order between $\lb{b}$
and $\lb{c}$ or the order between $\lb{e}$ and $\lb{f}$. However, by
breaking the order between $\lb{b}$ and $\lb{c}$ the value of
$\ell=2$, assigned at $\lb{a}$, no longer flows to $\lb{b}$ and hence
$\vr{y}$ is assigned the wrong value 1. Similarly by breaking the
order $\lb{e}$ and $\lb{f}$ the value of $\vr{m}=3$, assigned at
$\lb{d}$ no longer flows to $\lb{e}$ and hence $\vr{x}$ is assigned
the wrong value 1. In a nutshell, it is not possible to create a
symbolic trace whose execution will yield the state where
$\ell=\vr{m}=0$, $\vr{x}=5$, and $\vr{y}=3$. Notice that the problem
appeared because of the use of the same local variable in two
definitions. Such a scenario is unavoidable  when
(i) multiple reads are assigned to the same local variable, and/or
(ii) in the case of loops the local variable appears in a write
instruction within the loop.

We propose to handle such cases by renaming local variables,
viz. $\vr{\ell}$ and $\vr{m}$ in this case. For example, the execution
of trace
$\sigma=\Assign{\ell}{2}.\fcolorbox{gray}{gray}{$\Assign{\ell_1}{\ell}$}\lbr.\Assign{\ell}{x}\lbr.\fcolorbox{gray}{gray}{$\Assign{y}{\ell_1+1}$}\lbr.\Assign{m}{3}\lbr.\fcolorbox{gray}{gray}{$\Assign{m_1}{m}$}\lbr.\Assign{m}{y}\lbr.\fcolorbox{gray}{gray}{$\Assign{x}{m_1+2}$}$
results in state $\ell=\vr{m}=0$, $\vr{x}=5,\vr{y}=3$ as required by a
TSO execution. Let us look at instructions highlighted in gray color
more carefully. We earlier saw that the problem arises when reordering
$\lb{b}-\lb{c}$ and $\lb{e}-\lb{f}$ instructions as their reordering
will break the value flows of $\ell$ and $\vr{m}$ from $\lb{a}$ and
$\lb{d}$ respectively. Therefore, we create new instances of these
local variables, $\ell_1$ and $\vr{m_1}$, to take the \emph{snapshot}
of $\ell$ and $\vr{m}$ respectively which are later used in the write
instructions $\lb{b}$ and $\lb{e}$. This renaming ensures that even if
we reorder $\lb{b}-\lb{c}$ and $\lb{e}-\lb{f}$ instructions (as done
in  $\sigma$) the correct value flows from $\lb{a}$ to $\lb{b}$ and
from $\lb{d}$ to $\lb{e}$ are not broken. We will show that for a
buffer bound of $k$ it is sufficient to use at most $k$ instances of
these local variables and they can be safely reused even in the case of loops. We call such symbolic traces, that correspond to $\TSO_{k}$ executions, as \emph{SC interpretable traces}. Formally, 
 \emph{SC interpretation} of a trace $\sigma \in  \Labl^*$ is a function $\SCI:\Labl^* \times \Var \to \Val \cup \{\und\}$.
such that $\SCI(\sigma,\vr{x})$  calculates the last value assigned
to variable $\vr{x}$ in the sequential execution of $\sigma$. For example, if $\sigma=\lb{a.b.c}$ where labels $\lb{a},\lb{b}$, and $\lb{c}$ denote $\ell:=3$, $\vr{x}:=\ell+2$ and $\vr{y}:=2$ respectively then $\SCI(\sigma,\vr{x})=5$ and $\SCI(\sigma,\ell)=3$. Label $\und$ is used to denote the in-feasibility of $\sigma$ as some boolean expressions in $\mathsf{assume}$ instructions may become unsatisfiable because of the values that flow in them. If $\sigma$ does not contain any assignment to $\vr{x}$ then $\SCI(\sigma,\vr{x})$ returns the initial value of $\vr{x}$. 

Let us now construct a transition system such that the traces of this transition system represent SC interpretable traces corresponding to $\TSO_{k}$ semantics. We represent this transition system as $\TSOr{k}=\langle \Str,\tor{k},\str_0 \rangle$. 
Every state $\str \in \Str$ is of the form $(\Q,\LI,\Buffr_k)$ such
that $\Q:\TID\to Q$ represents process control states, and
$\Buffr_k:\TID\to (\SV \times \Labl)^k$ represents per process buffers
of length $k$. Unlike the buffers of $\TSO_k$, these buffers contain write instruction labels along with the modified shared variable. A function $\LI:\TID \times \LV \to \mathbb{N}$ tracks the instances of the local variables which have been used (for renaming purposes) in the construction of traces up to a given state. 
\vspace{-0.7cm}
\begin{figure}
\begin{center}
\removelatexerror
\scalebox{0.85}{\parbox{\textwidth}{%
\begin{minipage}{0.5\linewidth}
\begin{align}\tag{\textsc{MRead}$^{\sharp}$}\label{tag:modmread}
\frac
{\begin{array}{c}\Inst(a)=(\ell:=x), \\~ \restr{\Buffr_k(t)}{\{x\}\times \Labl}=\epsilon\end{array}}
{(\Q,\LI,\Buffr_k)\stackrel{a}{\tor{k}}(\Q',\LI,\Buffr_k)}
\end{align}	
\end{minipage}
\hspace{1cm}
\begin{minipage}{0.5\textwidth}
\begin{align}\tag{\textsc{LWrite}$^{\sharp}$}\label{tag:modmlwrite}
\frac
{\begin{array}{c}\Inst(a)=(\ell:=e)\end{array}}
{(\Q,\LI,\Buffr_k)\stackrel{a}{\tor{k}}(\Q',\LI,\Buffr_k)}
\end{align}	
\end{minipage}

\begin{minipage}{0.5\linewidth}
\begin{align}\tag{\textsc{Assume}$^{\sharp}$}\label{tag:modassume}
\frac
{\begin{array}{c}\Inst(a)=(\assume{e})\end{array}}
{(\Q,\LI,\Buffr_k)\stackrel{a}{\tor{k}}(\Q',\LI,\Buffr_k)}
\end{align}
\end{minipage}
\hspace{1cm}
\begin{minipage}{0.5\textwidth}
\begin{align}\tag{\textsc{flush}$^{\sharp}$}\label{tag:modflush}
\frac
{\begin{array}{c}\Buffr_k=(x,a).\Buffr'_k\end{array}}
{(\Q,\LI,\Buffr_k)\stackrel{a}{\tor{k}}(\Q',\LI,\Buffr'_k)}
\end{align}
\end{minipage}
\caption{All rules assume transitions for thread $t$, ie. $\Q[t]=q$, $(q,a,q')\in \delta_t$, and $\Q'=\Q[t\leftarrow q']$}}}
\label{fig:tsok1}
\end{center}
\end{figure}
\vspace{-1cm}
First, we define $\tor{k}$ for simple cases, viz. read from memory,
operations associated with local variables like $\assume{e}$ and
$\ell:=e$, and non-deterministic flush. In Rules
\ref{tag:modmread}, \ref{tag:modmlwrite}, and \ref{tag:modassume} the
labels that denote these operations are put in the trace with only
change in the control state of the process. As there is no notion of
local and global valuation in a state $\str$ of the transition system,
no update takes place unlike in $\TSO_{k}$. For memory read operation,
in Rule \ref{tag:modmread}, the condition on the buffer of $\Proc{t}$
is the same as in $\TSO_{k}$. For non-deterministic flush operation,
Rule \ref{tag:modflush} removes the first label present in the buffer
of $\Proc{t}$ and puts that in the trace. In rule \ref{tag:modassume},
the assume instruction is simply put in the trace without evaluating
the satisfiability of the boolean expression. This is different from
the corresponding rule in $\TSO_{k}$. This difference follows from
the fact that we are only interested in constructing symbolic
traces. Symbolic model checking of these traces will ensure that only
feasible executions get analyzed (where all assume instructions hold
true).
Now let us look at the remaining three operations, viz. read from the buffer, write to the buffer and fence instruction, in detail.
\begin{description}
 \item[\textsc{Buffered Read}] Like $\TSO_k$, this transition takes place when $\Proc{t}$ executes an instruction $\ell:=\vr{x}$ to read the value of shared variable $\vr{x}$ and store it in its local variable $\ell$. 
 
\scalebox{0.8}{\parbox{\linewidth}{%
\[\tag{\textsc{BRead}$^{\sharp}$}\label{tag:modbread}
\frac
{\begin{array}{c}\Inst(a)=(\ell:=x), \restr{\Buffr_k}{\{x\}\times \Labl}=\alpha.(x,b),\\
\Inst(b)=(x:=e), \Inst(c)=(\ell:=e)
\end{array}}
{(\Q,\LI,\Buffr_k)\stackrel{c}{\tor{k}}(\Q',\LI,\Buffr_k)}
\]}}

 For this transition to take place, the buffer of $\Proc{t}$ must have at least one write instruction that modifies the shared variable $x$. Conditions $\restr{\Buffr_k}{\{x\}\times \Labl}=\alpha.(x,b)$ and $\Inst(b)=(x:=e)$ ensure that the last write to $x$ in $\Buffr_k$ of $\Proc{t}$ is due to instruction $\Inst(b)$ which is of the form $(\vr{x}:=e)$. Under these conditions, in $\TSO_k$, read of $x$ uses the value of expression $e$ to modify $\ell$. Whereas in $\TSOr{k}$ a label $c$ is added to the trace such that $\Inst(c)$ represents the assignment of $e$ to variable $\ell$.
 
 \item[\textsc{Buffered Write}] This transition takes place when $\Proc{t}$ executes a write instruction of the form $x:=e$. Let $\vv{\ell}$ be a set of local variables used in expression $e$. For each of the local variables $\ell$ in $\vv{\ell}$, an integer $\LI(\ell)$ is used to create an assignment instruction of the form $\ell_{\LI(\ell)}:=\ell$. These instructions are put in the trace (through corresponding symbolic labels $\vv{a_{\ell}}$). Further, expression $e$ is also modified where every instance of a local variable $\ell$ in $\vv{\ell}$ is substituted with $\ell_{\LI(\ell)}$. 
 
\scalebox{0.8}{\parbox{\linewidth}{%
\[\tag{\textsc{BWrite}$^{\sharp}$}\label{tag:modbwrite}
\frac
{\begin{array}{c}\Inst(a)=(sv:=e), \FV(e)=\vv{\ell}, \len{\Buffr_k}<k,\\
\forall \ell\in \vv{\ell},~ \mbox{ create a label $a_{\ell}$ (if not already present in $\Labl$)} s.t.\\~~~~~\Inst(a_{\ell})=(\ell_{\LI(\ell)}:=\ell), \LI'[\ell]=\LI[\ell]\%(k+1)+1\\
\mbox{ create a label $a'$ (if not already present in $\Labl$)} s.t\\
~~~~~\Inst(a')=(sv:=e'), e'=\subst{e}{\vv{\ell}}{\vv{\ell_{\LI(\ell)}}},\\
\Buffr_k'=\Buffr_k[t\leftarrow \Buffr_k[t].(sv,a')]
\end{array}}
{(\Q,\LI,\Buffr_k)\stackrel{\vv{a_{\ell}}}{\tor{k}}(\Q',\LI',\Buffr'_k)}
\]}}
 
 This modified expression $e'$ is denoted $\subst{e}{\vv{\ell}}{\vv{\ell_{\LI(\ell)}}}$ in Rule \ref{tag:modbwrite}. A label, $a'$, representing the assignment of $e'$ to $x$ is put in the buffer in the form of a tuple $(x,a')$.
 Note that the transition rule \ref{tag:modbwrite} increases the value of $\LI(\ell)$ (modulo $(k+1)$) for every local variable $\ell$ present in expression $e$. We can show the following property,
\begin{lemma}\label{lem:tso:wrap}
 For a state $\str=(\Q,\LI,\Buffr_k)$ of $\TSOr{k}$, if $\LI(\ell)=m$ then local variable $\ell_m$ does not appear in any write instruction used in buffers $\Buffr_k$.
\end{lemma}
\begin{proof}
 Suppose $\LI(\ell)=m$ holds. By assumption, local variables among
 processes are disjoint therefore the only possibility is that
 $\Buffr_k[t]$ contains a write instruction that uses local variable
 $\ell_m$. If this were the case then there must be at least $k+1$
 different writes appearing between that write and the time $\str$ is
 reached. This holds because every write, that uses a local variable
 $\ell$ first increments its index by 1 and wraps around after
 $k+1$. This incremented index is then used to create an instance of
 the local variable $\ell$ used in this write operation. But it
 contradicts our assumption that the buffer is of bounded length
 $k$.
\end{proof}
The above lemma is used in the equivalence proof of $\TSO_k$ and $\TSOr{k}$.
  \item[\textsc{Fence}] Fence instruction, like $\TSO_k$, gets enabled only when $\Buffr_k[t]$ is empty. In the resultant state, function $\LI(t,\ell)$ is set to 1 for every local variable $\ell$ of Process $\Proc{t}$. This enables the reuse of indices in Function $\LI$ while preserving Lemma \ref{lem:tso:wrap}.
  
  \scalebox{0.8}{\parbox{\linewidth}{\begin{align}\tag{\textsc{Fence}$^{\sharp}$}\label{tag:modbfence}
\frac
{\begin{array}{c}\Inst(a)=(\fence), \Buffr_k[t]=\epsilon\\
\LI'=\LI[(t,\ell)\leftarrow 1], \forall \ell\in \LV_t
\end{array}}
{(\Q,\LI,\Buffr_k)\stackrel{\epsilon}{\tor{k}}(\Q',\LI',\Buffr_k)}
\end{align}	}}\\

\end{description}
%
To show the equivalence of $\TSO_k$ and $\TSOr{k}$ we want to prove
the following; (i) for every state $\st$ reachable in $\TSO_k$ there
exists a trace $\sigmar$ in $\TSOr{k}$ such that the \emph{SC
  interpretation} of $\sigmar$ reaches a state with the same global
memory and local memory as of $\st$, and (ii) for every trace
$\sigmar$ of $\TSOr{k}$ such that its \emph{SC interpretation} is not
$\und$ (i.e. execution should be feasible) there exists a state $\st\in
\TSO_k$ with same global and local memory as obtained after the
\emph{SC interpretation} of $\sigmar$. We formally prove the following
theorem in the Appendix.
\begin{theorem}\label{thm:eq}
 Transition systems $\TSO_k$ and $\TSOr{k}$ are equivalent in terms of
 state reachability. 
\end{theorem}
In Theorem \ref{lem:fp} we used the restricted set $\nrestr{\St}{\Q,\GM,\LM,\Buff_{lst}}$ as a means to define fixed point. However, there are no explicit representations of the global memory ($\GM$) and the local memory ($\LM$) in the state definition of $\TSOr{k}$. Therefore, in order to define a fixed point condition like Theorem \ref{lem:fp} we first augment the definition of state in $\TSOr{k}$ to include global memory and local memory. Let $\GMr:\SV\to \Lab$ and $\LMr:\TID\times \LV \to \Lab$ be the functions assigning labels (of write instructions) to shared variables and local variables respectively. Specifically, $\GMr(\vr{x})=\lb{a}$ means that the write instruction at label $\lb{a}$ was used to define the current value of $\vr{x}$ in this state. Similarly, $\LMr(t,\vr{\ell})=\lb{a}$ means that the write instruction at label $\lb{a}$ was used to define the current value of local variable $\ell$ of process $t$. Note that in the construction of $\TSOr{k}$ the values written by these write instructions are only being represented symbolically using instruction labels. Therefore we need a way to relate the instruction labels and the actual values written. For a concurrent program $P$ with finite data domain it is possible to construct an equivalent program $P'$ such that every assignment to variables in $P'$ is only of constant values.
For example, consider the program in Figure \ref{fig:orig} such that the domain of variable $\vr{x}$ is $\set{1,2}$. This program is equivalent to the program in Figure \ref{fig:mod} where only constant values are used in the write instructions. Here the domain of $\vr{x}$ is used along with if-then-else conditions to decide the value that needs to be written to $\vr{y}$.\\
%
\begin{wrapfigure}[8]{l}{0.3\linewidth}
  \begin{minipage}[b]{0.3\linewidth}
   \centering{
  $\begin{array}{l}
  \Assign{\ell}{x}\\
  \Assign{y}{\ell+3}
  \end{array}$
   \captionof{figure}{}
 \label{fig:orig}}
\end{minipage}\hspace{0.25in}
\begin{minipage}[b]{0.3\linewidth}
   \centering{
$\begin{array}{l}
  \Assign{\ell}{x}\\
  \mathsf{if}(\ell=1)\\
  ~~\Assign{y}{4}\\
  \mathsf{else~if}(\ell=2)\\
  ~~\Assign{y}{5}\\
 \end{array}$ 
 \captionof{figure}{}
 \label{fig:mod}}
\end{minipage}
\end{wrapfigure}

After this transformation, every write label uniquely identifies the value written to a shared variable. Hence the functions $\GMr$, $\LMr$ can be extended to $\SV\to \Val$ and $\LV\to \Val$ respectively. This allows us to use Theorem \ref{lem:fp} for checking the fixed point.

%% file: fence.tex
\newcommand{\Cmpt}{\mathsf{Cmpt}} \newcommand{\po}{\mathrm{po}}
\newcommand{\ppo}{\mathrm{ppo}} \newcommand{\delay}{\mathrm{Dlay}}
\subsection{Fence Insertion For Program Correction}\label{subsec:fence}
Let $P$ be a program that is correct under the SC memory model. Let
$\sigma$ be an execution of $P$ that violates the given safety
property under the TSO memory model. We can insert a $\fence$
instruction in $P$ so that $\sigma$ does not appear as an execution
under the TSO memory model. Towards this we use the \emph{critical
  cycle} based approach of \cite{ShashaSnirtoplas88} and
\cite{aglavestabilitycav11} to detect the locations of $\fence$
insertions. For an execution $\sigma$, let $\Cmpt_{\sigma}$ be a
\emph{competing}\cite{aglavestabilitycav11} or
\emph{conflicting}\cite{ShashaSnirtoplas88} relation on the read and
write events of $\sigma$ such that $(a,b)\in \Cmpt_{\sigma}$ iff (i)
both memory events operate on the same location but originate from
different processes, (ii) at least one of them is a write instruction,
and (iii) $a$ appears before $b$ in $\sigma$. Let $\po_{\sigma}$
denote the program order among instructions of processes present in
$\sigma$. This is defined based on the process specification. Let
$\ppo_{\sigma}=\po_{\sigma}\setminus \{(a,b)\mid a\in W,b\in R,
(a,b)\in \po_{\sigma}\}$ be a subset of $\po_{\sigma}$ preserved under
TSO memory model, i.e. everything except write-read orders.

An Execution $\sigma$ contains a critical cycle
$\stackrel{cs}{\to}\subseteq (\Cmpt_{\sigma} \cup \po_{\sigma})^+$ iff
(i) no cycle exists in $(\Cmpt_{\sigma} \cup \ppo_{\sigma})^+$, (ii)
per process there are at most two memory accesses $a$ and $b$ in
$\stackrel{cs}{\to}$ such that $\loc(a)\neq \loc(b)$, and (ii) for a
given shared variable $x$ there are at most three memory accesses on
$x$ which must originate from different processes.
Following Theorem 1 of \cite{aglavestabilitycav11}, an execution in
TSO is sequentially consistent if and only if it does not contain any
critical cycle. Therefore, in order to forbid an execution in TSO that
is not sequentially consistent, it is sufficient to ensure that no
critical cycle exists in that execution.
To avoid critical cycles, we need to strengthen the $\ppo_{\sigma}$
relation by adding a minimal set of program orders such that Point (i)
of critical cycle definition is not satisfied, i.e. finding a set
$\delay\subseteq \po_{\sigma}\setminus \ppo_{\sigma}$, set of
write-read pairs of instructions within each process, such that
$(\Cmpt_{\sigma} \cup \ppo_{\sigma}\cup \delay)^+$ becomes
cyclic. Once we identify that minimal set of program orders we insert
$\fence$ instructions in between them to enforce the required
orderings.
 

\paragraph{Overall Algorithm}
Algorithm that combines incremental buffer bounded verification and
fence insertion for finite data programs works as follows. We start
the verification with buffer bound of 0. Towards this, the transition
system $\TSOr{k}$ is constructed using the relation $\tor{k}$ given in
this section. This transition system is represented as an automaton
with error location representing the accepting states and initial
locations representing the initial state. The set of traces accepted
by this automaton are the passed to the trace partitioning algorithm
implemented by \cite{DBLP:journals/corr/NarayanGA15} in the tool
\texttt{ProofTraPar}. If an erroneous trace is found then the program
is not safe even under the SC memory model and hence the algorithm
returns the result as `Unsafe'. If all traces satisfy the given safety
property then the bound is increased by one and the analysis starts
again. If an error trace is found for non-zero buffer bound then the
critical cycles are obtained from this trace. Using these critical
cycles a set of fence locations are generated and the input program is
modified by inserting fences in the code. After the modification the
analysis again starts with the same bound. This is just an
implementation choice because even if we increase the bound after the
modification still the fixed point will be eventually reached.

%% file: experiments.tex
\newcommand{\expeter}[1]{{\em {peterson}}}
\newcommand{\exszy}[1]{{\em szymanksi}}
\newcommand{\exabp}[1]{{\em ABP}}
\section{Experimental Results}\label{sec:experiment}

{{
\begin{figure}[h]
{
\centering
\begin{tabular}{l|c|c|c|c|c | c}
Program & \# P& \multicolumn{2}{c}{ProofTraPar} & \multicolumn{2}{c}{Memorax\cite{memoraxtacas11}} & \# F  \\
& & Time & Memory & Time & Memory \\
& & (Sec)& (MB) & (Sec) & (MB)\\
\hline& & & &   &   \\
Peterson.safe & 2& 1.19 & $\mathbf{20}$ & $\mathbf{0.9}$ & 43 & 2\\
Dekker.safe & 2& $\mathbf{1.6}$ & $\mathbf{21.3}$ & 54.2 & 676 & 2\\
Lamport.safe & 2& $\mathbf{17}$ & $\mathbf{42}$ & 97 & 2312 & 4\\
Szymanksi.safe & 2 & 27 & 121 & $\textbf{ERR}$ & $\textbf{ERR}$ & 4\\
Alternating Bit(ABP) & 2& 3.12 & 39 & $\mathbf{0.17}$ & $\mathbf{11}$ & 0\\
Dijkstra & 2& $\mathbf{16}$ & $\mathbf{70}$ & - & - & 2\\
Pgsql & 2& $\mathbf{1.2}$ & $\mathbf{20}$ & 210 & 2800 & 2\\
RWLock.safe (2R,1W) & 3&  $\mathbf{41}$ & $\mathbf{164}$ & - & - & 2\\
clh & 2& $\mathbf{326}$ & $\mathbf{1500}$ & - & -  & 0\\
Simple-dekker & 3& $\mathbf{103}$ & $\mathbf{155}$ & 600 & 3280 & 3\\
Qrcu.safe (2R,1W) & 3& $\mathbf{490}$ & $\mathbf{3000}$ & - & - &  0 \\ \\
\end{tabular}
\caption{Comparison of our tool with Memorax\cite{memoraxtacas11}. Time out, denote by `-' is set to 10 minute. \#P and \#F denote number of processes and number of fences synthesized.}
\label{tab:comparison}
}
\vspace{-0.5cm}
\end{figure}

}}

We implemented our approach by extending the tool \texttt{ProofTraPar}
which implements the trace partitioning based approach of
\cite{DBLP:journals/corr/NarayanGA15}. We implemented $\TSOr{k}$
semantics and fixed point reachability check on top of
\texttt{ProofTraPar}. Its performance was compared against
\texttt{memorax} which implements sound and complete verification of
state reachability under the TSO memory model. Note that other tools
which exist in this landscape of relaxed memory verification either
consider SC behaviour as specification
\cite{persistenceesop15,aglavestabilitycav11,bouajjanirobustnessesop13}
or are sound but not complete
\cite{remmextacas13,DBLP:conf/tacas/AbdullaAAJLS15,Zhang:2015:DPO:2737924.2737956}. However
\texttt{memorax} does not assume any bound on the buffer size and it
uses the coverability based approach of well-structured-transition
systems. Table \ref{tab:comparison} compares the performance, in terms
of time and memory, of our approach with \texttt{memorax}. We ran all
experiments on Intel i7-3.1GHz, 4 core machine with 8GB RAM. Out of 11
examples, our tool outperformed \texttt{memorax} in 8
examples. Our tool not only performed better in terms of time but also
in terms of the memory consumption. Except in two cases, \emph{qrcu}
and \emph{clh queue}, on every other example our tool consumed less
than 200 MB of RAM. Whereas \texttt{memorax} in most cases took more
than 500 MB of RAM and in some cases even touched the 3GB
mark. Programs like \emph{Alternating bit protocol}, \textit{clh queue} and \emph{Qrcu}(quick read copy update algorithm) remain
correct even under TSO memory model. For other algorithms where bugs
were exposed under TSO we were able to synthesize 
fences to correct their behaviour.

\paragraph{Analysis of the benchmarks} 
\texttt{memorax} performed better on three benchmarks, viz.
\expeter{}, \exszy{}, and \exabp{}. After carefully looking at them we
realized that the performance of \texttt{memorax} loosely depends upon
the number of backward control flow paths from error location to the
start location, and number of write instructions present along those
paths. In benchmarks where \texttt{ProofTraPar} outperformed
\texttt{memorax}, viz. \emph{dekker}, \emph{lamport}, \emph{clh},
\emph{qrcu}, more than two such control paths exist. To further check
this hypothesis experimentally we modified \expeter{} and \exabp{} to
add a write instruction along an already existing control flow path
where no write instruction was present. This write was performed on a
variable which was never read and hence did not affect the
program. After this modification \texttt{memorax} became more than 6
time slower in analyzing these two benchmarks. Further, the analysis
of these modified benchmarks with \texttt{ProofTraPar} exhibited a
very little (less than a second) increase in time as compared to the
unmodified benchmarks. Interestingly, a bug was exposed in
\texttt{memorax} when we made a similar change in \exszy{}. As a
result of this bug the modified program \exszy{} was declared as
safe. Note that the original program \exszy{} is incorrect under TSO
and we only modified the code by adding a write instruction to an
unused variable. Therefore it is not possible for the modified program
\exszy{} to become safe unless there is a bug in the tool. This bug
was also confirmed by the author of \texttt{memorax}. 

\vspace{-0.3cm}
\subsection{Discussion}
Note that \texttt{memorax} starts from the symbolic representation of
all possible configurations of buffer contents which it further refines
using backward reachability analysis. However, in our approach we
start from a finite and small buffer bound ( an under-approximation)
and keep expanding until we reach a fix point. We believe that this
difference, picking an over-approximation as a starting point in one case and an
under-approximation as a starting point in the other case, plays a crucial role 
in the better performance of our approach on these benchmarks.

In all benchmarks, except \expeter{}, buffer size of 1 was
sufficient to expose the error. In \expeter{}, buffer size of 2 was
needed to expose the bug. Effectively, the buffer size depends upon the
minimum distance (along control flow path) between a write and a read
instruction within a process whose reordering reveals the bug. In the case
of \expeter{}, this distance is 2 since the reordering of first
instruction (write to $flag_i$) and third instruction (read of
$flag_j$) within each process reveals the bug. In our benchmarks
fence instructions were inserted after finding an erroneous trace, as
discussed in Section \ref{subsec:fence}. Fence instruction restricts
the unbounded growth of the buffer by flushing the buffer contents. As
a result, when a fence is inserted within a loop the buffer never grows
in size with loop iterations and fix point is reached quickly. In fact,
for all the benchmarks, if a bug was exposed with buffer size $k$ then
after inserting the fence instruction the fix point was reached with
buffer size $k+1$. Benchmarks which remain correct under TSO, a
larger buffer bound was required to reach the fix point and this bound
depends upon the number of write operations in each process. As a
result, their analysis took longer time and consumed more
memory. Detailed analysis of the benchmarks and the tool are available
at \url{www.cse.iitd.ac.in/~chinmay/ProofTraParTSO}.

%% file: conclusion.tex
\section{Conclusion and Future Work}\label{sec:conclusion}
This paper uses the trace partitioning based approach to verify state
reachability of concurrent programs under the TSO memory model. We
have also shown that for finite state programs there exists a buffer
bound such that if program is safe up-to that bound then the program is
guaranteed to be safe for unbounded buffers as well. This work can be
easily extended to PSO memory model as well. This method gives us an
alternate decidability proof of state reachability under TSO (and PSO)
memory model. We have also shown experimentally that for standard
benchmarks used in the literature such a bound is very small (in the
range of 2-4) and hence we may use SC verification based methods to
efficiently check concurrent programs under these memory models. 
We believe that for other buffer based memory models a buffer bound
can be shown to exist in a similar manner. 
Recently \cite{Lahav:2016:TRC:2837614.2837643} proposed a buffer based
operational semantics for C11 model. It will be interesting to
investigate the use of bounded buffer based method proposed in this
paper to that semantics as well.